\documentclass[preprint,authoryear,11pt]{elsarticle}

\usepackage[T1]{fontenc}
\usepackage[utf8]{inputenc}
\usepackage[british]{babel}
\usepackage{lmodern}
\usepackage{microtype}
\usepackage{amsmath,amssymb,amsthm}
\usepackage{bm}
\usepackage{mathtools}
\usepackage{graphicx}
\usepackage{booktabs}
\usepackage{tikz}
\usepackage{enumitem}
\usepackage[hidelinks]{hyperref}

\theoremstyle{plain}
\newtheorem{theorem}{Theorem}
\newtheorem{proposition}[theorem]{Proposition}
\newtheorem{lemma}[theorem]{Lemma}

\theoremstyle{definition}
\newtheorem{definition}[theorem]{Definition}
\newtheorem{principle}[theorem]{Principle}
\theoremstyle{remark}
\newtheorem*{remark}{Remark}

\newcommand{\E}{\mathbb{E}}
\newcommand{\R}{\mathbb{R}}
\newcommand{\entropy}{H}
\newcommand{\entrate}{\bar{H}}
\newcommand{\relent}{D}

\newcommand{\one}{\mathbf{1}}
\newcommand{\T}{\widehat{T}}

\journal{Journal of Mathematical Economics}

\begin{document}

\begin{frontmatter}

\title{Equilibrium as a Limit: The Competitive Canon Nested in an Adaptive, Information-Theoretic Economy}

\author{Avishek Bhandari}
\address{School of Humanities, Social Sciences and Management, Indian Institute of Technology Bhubaneswar, India}
\ead{avishekb@iitbbs.ac.in}

\begin{abstract}
\noindent The competitive equilibrium of general-equilibrium theory exists as a fixed point and is, by the theory's own results on aggregate excess demand, in general silent on whether that fixed point is unique, stable, or attained. This paper takes the economy to be not a configuration to be solved for but a process to be recovered: an asymptotically mean stationary information source carrying a partially identified operator of statistical dependence, populated by agents that are finite-capacity information channels. Within this adaptive order the competitive, rational-expectations equilibrium is recovered exactly, as a joint limit taken along an explicit scaling path. Three parameter limits and two fixed-point conditions deliver it: the entropy rate falls to zero, agent channel capacity diverges, selection intensity grows infinitely sharp, adaptive learning reaches its expectationally stable rest point, and the recovered structure ceases to co-evolve. At that corner the limiting object satisfies the axioms of the canon and its rest state is a Walrasian equilibrium; away from it the adaptive economy is a strict generalisation, carrying a positive entropy rate and a recovered dependence structure that the equilibrium primitive cannot express. We give the nesting as a theorem, establish the result-by-result correspondence with existence, with the Sonnenschein--Mantel--Debreu indeterminacy, and with the regular-economies recovery, and characterise exactly what the equilibrium limit erases: novelty, non-stationarity, irreversibility, adaptation, and the moving structure of crisis. A behaviour-freedom result shows that the recovery requires nothing of the rationality of the agents who generate the path, so the same primitive houses human and machine decision-makers; rationality occupies a vanishing slice of the preference relations an agent may instantiate. The normative standard is re-founded in step, from the static efficiency of an allocation to the viability of a trajectory, with competitive efficiency recovered as its frictionless corner. Full proofs are given in the appendix.

\end{abstract}


\end{frontmatter}

\section{Introduction}
\label{sec:intro}

The mathematics an economics chooses is never neutral, for it encodes a prior decision about what an economy most fundamentally is. The science inherited a particular answer in its most beautiful form. From Walras through Arrow and Debreu the economy was taken to be, at bottom, a system of simultaneous relations whose solution is an equilibrium, a configuration of prices at which every plan is mutually consistent; and the summit of that tradition, the axiomatic theory of value, established under stated conditions that such a configuration exists \citep{arrowdebreu1954,mckenzie1954,debreu1959}. The accomplishment is not contested here. What is contested is the identification of the economy with the fixed point, for the same tradition proved, in its theorems on the structure of aggregate excess demand, that the equilibrium it had shown to exist it could not in general show to be unique, stable, or reached \citep{sonnenschein1972,mantel1974,debreu1974}. A primitive mute on exactly the questions a science of crises, of growth, and of adaptation must address is, for those purposes, the wrong primitive to build upon.

This paper relocates the rigour rather than abandoning it. The object a proof certifies is no longer the existence of a fixed point but the recovery of a structure from the observable motion of an economy. The economy is modelled as an asymptotically mean stationary information source, a law over entire histories possessing a well-defined entropy rate, that carries a directed operator of statistical dependence recovered, only in part, from the realised path \citep{gray2023,graykieffer1980}. Its agents are not unconstrained optimisers but information channels of finite capacity, whose choice under a constraint on bits is a Gibbs law over actions. The change is a change in what one proves theorems about: from a configuration the economy is assumed to occupy to a process it is observed to perform.

A re-founding earns its name only if it contains its predecessor. The central result of this paper is that it does, and exactly. We show that the competitive, rational-expectations equilibrium of the canon is recovered as a \emph{joint limit}, taken along an explicit scaling path, of the adaptive economy (Theorem~\ref{thm:nesting}). The limit is reached when five conditions hold together: the entropy rate falls to zero; the channel capacity of every agent diverges; the selection intensity of the population dynamics grows infinitely sharp; adaptive learning converges to its expectationally stable fixed point; and the recovered structure ceases to co-evolve with the strategies that act upon it. At that corner the limiting object is a competitive economy with rational expectations on a fixed structure, its rest state is a Walrasian equilibrium, and the axioms of the canon hold there. Away from the corner the adaptive economy is a \emph{strict} generalisation: it carries a positive entropy rate and a non-degenerate, co-evolving dependence operator, and neither has any expression in the equilibrium primitive, which knows only the fixed point.

Three features distinguish this nesting from the familiar observation that rational choice is a limit of noisy choice. First, the limit is \emph{joint and singular}: each of the five conditions annihilates a distinct quantity the adaptive primitive carries, and we enumerate the five erased contents exactly (Proposition~\ref{prop:erasure}). The equilibrium primitive is thereby exhibited not as a neutral simplification but as the deletion of novelty, non-stationarity, irreversibility, adaptation, and the moving structure of crisis. Second, the correspondence with the canon is \emph{result by result}: existence is inherited, the Sonnenschein--Mantel--Debreu indeterminacy is shown to be a property of the frictionless corner rather than of the economy that the corner approximates, and the generic local determinacy of regular-economy theory has an exact counterpart in the determinacy of the learning dynamics (Theorem~\ref{thm:nesting}, clauses (a)--(c)). Third, the recovery is \emph{behaviour-free}: it requires nothing of the completeness, transitivity, or determinacy of the agents who generate the path (Proposition~\ref{prop:behaviourfree}). The primitive is therefore indifferent to the substrate of choice, housing a human deciding by heuristic and a machine executing a stochastic policy as the same object, and rationality is shown to occupy a vanishing slice of the preference relations an agent may instantiate (Proposition~\ref{prop:rare}).

The argument is organised as a small axiomatic system, stated so that the contrast with the canon is exact at every step. Section~\ref{sec:primitive} fixes the formal objects and the four principles of the adaptive order. Section~\ref{sec:agent} develops the agent and proves the three results that detach the recovery from rationality. Section~\ref{sec:nesting} states and proves the general nesting theorem and the result-by-result correspondence. Section~\ref{sec:erasure} proves the converse, what the limit costs. Section~\ref{sec:welfare} re-founds the normative standard as viability and recovers competitive efficiency as its frictionless corner. Section~\ref{sec:conclusion} concludes. Proofs are given in sketch in the main text and in fuller form in the appendix; the measure-theoretic and operator-theoretic apparatus on which they rest is that of the ergodic theory of asymptotically mean stationary sources \citep{gray2023}.

\paragraph{Related work.} The existence tradition and its internal terminus are the work of \citet{arrowdebreu1954}, \citet{mckenzie1954}, \citet{debreu1959}, the indeterminacy results of \citet{sonnenschein1972}, \citet{mantel1974}, \citet{debreu1974}, and the regular-economies recovery of \citet{debreu1970}. The bounded agent as an information channel draws on the rate-distortion and free-energy traditions; the stochastic-choice reading of the resulting Gibbs law is the random-utility model of \citet{mcfadden1974}, against whose regularity the irregular preference fields of \citet{thurstone1927}, \citet{luce1959}, and the wider stochastic-choice literature are defined. The non-rational agent draws on the documented limits of attention and judgement \citep{simon1955,tverskykahneman1974,kahnemantversky1979,tversky1969} and on the theory of incomplete preferences \citep{aumann1962,bewley2002}. The source-theoretic primitive is built on the ergodic theory of asymptotically mean stationary processes \citep{gray2023,graykieffer1980}.

\section{The adaptive primitive}
\label{sec:primitive}

We fix the formal objects in the measure-theoretic and operator-theoretic language to which the full construction is referred, then state the four principles that organise the theory.

\subsection{The source, the operator, and the agent}

The economy is a stochastic process $X=(X_t)_{t\in\mathbb Z}$ on a probability space $(\Omega,\mathcal F,\mathbb P)$ carried by a measurable shift $\tau$. It is \emph{asymptotically mean stationary} (AMS) when the time-averaged measures $n^{-1}\sum_{t=0}^{n-1}\mathbb P\circ\tau^{-t}$ converge setwise to a stationary mean $\bar\mu$. An AMS source possesses an \emph{entropy rate}
\begin{equation}
\entrate=\lim_{n\to\infty}\frac1n\,\entropy(X_0,\dots,X_{n-1}),
\label{eq:rate}
\end{equation}
equal to the entropy rate of its stationary mean, the limit existing by the subadditivity of block entropies \citep{gray2023}. The rate is the per-period novelty the source generates: zero exactly when the present is asymptotically determined by the past, positive when the source is a genuine generator of the unforeseen.

The observables form the real Hilbert space $L^2_0(\bar\mu)$ of mean-zero, square-integrable functions of the path, with inner product the covariance. The dependence among observables is encoded in a bounded linear \emph{dependence operator} $\T$ on this space. Its self-adjoint part is diagonalised by the spectral theorem, $\T_{\mathrm s}=\int\lambda\,\mathrm dE(\lambda)$ for a projection-valued measure $E$, so that the structure of contemporaneous dependence is carried by a real spectrum and the directions attached to it; its skew part carries the directed, lag-bearing transmission. The resolvent $(I-z\T)^{-1}$, where it exists, sums the reverberations of a disturbance through the structure at successive removes.

The agent is fixed by an optimisation over these data. Given an observable $\xi\in L^2_0(\bar\mu)$ and a finite capacity $C$, the agent's representation is the channel that maximises expected attainable value subject to the information constraint $I(\xi;\,\text{response})\le C$, the rate-distortion problem of information theory \citep{coverthomas2006}. Introducing a multiplier $\beta$ for the constraint turns this into the unconstrained free-energy problem
\begin{equation}
\max_{q\in\mathcal P(\mathcal A)}\ \Big\{\,\E_{q}[U]-\tfrac1\beta\,\relent\!\big(q\,\big\Vert\,p_0\big)\,\Big\},
\label{eq:freeenergy}
\end{equation}
over probability measures $q$ on the action space $\mathcal A$, whose unique maximiser is the Gibbs law
\begin{equation}
q^{\ast}(a)\propto p_0(a)\,\exp\!\big(\beta\,U(a)\big).
\label{eq:gibbs}
\end{equation}
The multiplier $\beta$ is the inverse temperature, the shadow price of information. It indexes capacity at its two ends: as $\beta\to\infty$ the Gibbs law concentrates on the value-maximising action and the agent becomes the unconstrained optimiser; as $\beta\to0$ choice collapses onto the prior $p_0$ and the agent becomes pure habit. Expected value penalised by relative-entropy distance from a prior is the optimisation substrate of every agent in the theory, the maximum-entropy form familiar from statistical mechanics \citep{jaynes1957}.

\subsection{The four principles}

\begin{principle}[The adaptive primitive]
\label{p:prim}
The primitive of the theory is not a configuration to be solved for but a process to be recovered. An economy is an asymptotically mean stationary information source, possessing a well-defined entropy rate, that carries a directed operator of statistical dependence; the operator is recovered from the realised path of the economy and is identified only in part.
\end{principle}

\begin{principle}[The bounded agent]
\label{p:agent}
An agent is an information channel of finite capacity. The mutual information between the agent's environment and its response is bounded by that capacity, and subject to the bound the agent's choice is the free-energy optimum~\eqref{eq:gibbs}: a Gibbs law over actions, proportional to a prior weight times the exponential of scaled attainable value, whose concentration is governed by an inverse-temperature parameter that indexes the capacity.
\end{principle}

\begin{principle}[Co-evolution and selection]
\label{p:coevol}
Strategies and the structure of dependence co-evolve under selection, and the equilibrium concepts of the static theory are not posited but recovered as the converged fixed points of that evolution. A configuration of rational expectations is a rest point of adaptive learning, attained precisely when it is expectationally stable; a configuration of the population is retained precisely when it is evolutionarily stable.
\end{principle}

\begin{principle}[The certification gate]
\label{p:gate}
No quantity enters a claim or a policy except at the resolution at which the realised path identifies it. The recovered operator is identified only up to an observational-equivalence class: its spectrum and leading eigenvector are firmly identified; its resolvent, and the functionals built from it, are identified only up to a calibration; an individual directed edge lies beyond what the data identify. The boundary between the identified and the unidentified is governed by a signal-to-noise threshold, made precise in Proposition~\ref{prop:identification}. An admissible policy may steer only firmly identified functionals, and must refuse to act upon an unidentified edge.
\end{principle}

The four principles stand to the theorems of the adaptive order as the convexity and continuity axioms stood to the existence theorem of the canon. Principle~\ref{p:prim} replaces the price vector $p^{\ast}$, a static object whose existence is the theorem and whose motion is left to an adjustment story external to the definition, with a dynamic object whose entropy rate and recovered dependence are the measured content and whose rest state, not its motion, is the special case. Principle~\ref{p:agent} replaces the unconstrained maximiser with the infinite-capacity corner of a one-parameter family. Principle~\ref{p:coevol} replaces a posited equilibrium, whose stability is a separate and, by Sonnenschein--Mantel--Debreu, generally unanswerable question, with a configuration that the dynamics select. Principle~\ref{p:gate} is the structural discipline that makes recovery and intervention well posed, the counterpart in the present theory of the role convexity played in making existence well posed: where the equilibrium theory required the choice sets to be convex, the present theory requires every operative quantity to be certified. One caution governs every later use of the operator's spectrum: a statement about the determinacy of a dynamic equilibrium is never read off the recovered spectral radius, which certifies structure, not the stability of a forward-looking system.

Principle~\ref{p:gate} is not a posited restriction but a derived one. The identification it asserts, and the boundary it draws between what the realised path fixes and what it leaves free, is the content of the following proposition; the worked example of Section~\ref{sec:worked} exhibits it in closed form.

\begin{proposition}[Partial identification of the recovered operator]
\label{prop:identification}
Suppose the dependence operator is identified from the stationary mean $\bar\mu$ only through the second-order comovement of the observables, so that two operators are observationally equivalent when they induce the same contemporaneous covariance on $L^2_0(\bar\mu)$. Then:
\begin{enumerate}[label=\emph{(\roman*)},leftmargin=2.4em]
\item the self-adjoint part $\T_{\mathrm s}=\tfrac12(\T+\T^{\!\top})$ is identified, and with it the real spectrum $\{\lambda_k\}$ and the eigenvectors of $\T_{\mathrm s}$; in particular the leading eigenvalue $\lambda_{\max}$ and the leading eigenvector $v_{\mathrm{lead}}$ are firmly identified;
\item the skew part $\tfrac12(\T-\T^{\!\top})$, which carries the direction of an individual edge, is not identified: $\T$ and its transpose $\T^{\!\top}$ are observationally equivalent, sharing the self-adjoint part and reversing the skew part, and within any degenerate eigenspace of $\T_{\mathrm s}$ an orthogonal rotation produces a further observationally equivalent operator;
\item a resolvent $(I-z\T_{\mathrm s})^{-1}$ built on the identified leading eigenpair is identified only up to the calibration of $z$, and an individual directed edge $\T_{ij}$ is not identified;
\item from a sample of length $N$ in dimension $d$ the leading eigenpair is recoverable when the leading eigenvalue exceeds the threshold set by the noise level $\nu$, $\lambda_{\max}>\nu\sqrt{d/N}$, and is absorbed into the bulk of the spectrum, hence unidentified, below it.
\end{enumerate}
\end{proposition}

\begin{proof}[Proof sketch]
The contemporaneous covariance is a symmetric bilinear form on $L^2_0(\bar\mu)$ and determines, and is determined by, the self-adjoint part $\T_{\mathrm s}$; the spectral theorem then yields its real spectrum and eigenvectors, which is~(i). For~(ii), $(\T^{\!\top})_{\mathrm s}=\T_{\mathrm s}$ while $\T^{\!\top}$ reverses the skew part, so $\T$ and $\T^{\!\top}$ induce the same contemporaneous covariance yet exchange every off-diagonal pair $\T_{ij}\leftrightarrow\T_{ji}$; and if $Q$ is orthogonal and commutes with the spectral projections of $\T_{\mathrm s}$ then $Q\T_{\mathrm s}Q^{\!\top}=\T_{\mathrm s}$ while $Q$ moves the skew part, giving a continuum of equivalent operators on any degenerate eigenspace. Claim~(iii) is immediate from~(i) and~(ii). Claim~(iv) is the spike phase transition for the largest eigenvalue of a finite-rank perturbation of a sample covariance matrix \citep{baikbenarouspeche2005,benaychgeorgesnadakuditi2011}: above the threshold the leading empirical eigenpair is a consistent estimate of the population one, while below it the empirical eigenvalue is absorbed into the Marchenko--Pastur bulk and the empirical eigenvector is asymptotically uninformative.
\end{proof}

\section{The agent need not be rational}
\label{sec:agent}

The bounded agent of Principle~\ref{p:agent} wears the look of an optimiser, and a reader of the canon might take it for the rational maximiser fitted with a constraint. It is not. The Gibbs law~\eqref{eq:gibbs} is a \emph{stochastic} rule: at finite capacity the agent answers the same problem differently on different occasions, and what it optimises is a free energy, not a preference. This section makes the point general and exact, for it is the deepest break with the canon and the feature that detaches the recovery from any assumption on the agents.

\begin{definition}[The space of agents and the random preference field]
\label{def:agentspace}
Let $\mathcal X$ be a compact metric space of alternatives and $\mathcal R$ the space of binary relations on $\mathcal X$, with its Borel $\sigma$-algebra; write $\mathcal R^{\ast}\subset\mathcal R$ for the complete and transitive relations, the rational preorders. An \emph{agent} is a probability measure $\rho$ on $\mathcal R$, a \emph{random preference field}, together with a measurable rule resolving indifference; it induces a \emph{stochastic choice function} assigning to each finite menu $A\subseteq\mathcal X$ the law
\begin{equation*}
P_\rho(a\mid A)=\rho\big(\{\succsim\ :\ a\ \text{is}\ \succsim\text{-maximal in}\ A\}\big),\qquad a\in A,
\end{equation*}
equivalently, by a random-utility representation, the law of a random function $u\colon\mathcal X\to\R$ with $P_\rho(a\mid A)=\mathbb P\{u(a)\ge u(b)\ \forall b\in A\}$. The agent is \emph{rational} when $\rho=\delta_{\succsim^{\ast}}$ for a single $\succsim^{\ast}\in\mathcal R^{\ast}$ under a deterministic rule, in which case the stochastic choice function collapses to a single-valued demand.
\end{definition}

\begin{proposition}[The bounded agent is a random-utility agent]
\label{prop:randomutility}
The Gibbs law~\eqref{eq:gibbs} over a finite menu is the stochastic choice function of a random preference field: it is the law of choice when the systematic value $U$ is perturbed by independent extreme-value noise of scale $1/\beta$, the multinomial logit agent. The agent is therefore genuinely stochastic for every finite $\beta$ and deterministic only as $\beta\to\infty$. The logit field is the regular corner of the fields of Definition~\ref{def:agentspace}; the general field need not be regular, admitting the intransitive cycles and context-dependent reversals that the deterministic axioms forbid.
\end{proposition}

\begin{proof}[Proof sketch]
That $q^{\ast}(a)\propto p_0(a)e^{\beta U(a)}$ is precisely the probability that $a$ maximises $U(\cdot)+\varepsilon(\cdot)$ when the $\varepsilon(b)$ are independent Gumbel variates of scale $1/\beta$, the prior entering as an additive log-offset, is the logit identity \citep{mcfadden1974}; the variational derivation of $q^{\ast}$ from~\eqref{eq:freeenergy} is standard. Regularity of the logit field and the existence of random-utility representations for the regular stochastic choice functions are the substance of the stochastic-choice literature \citep{thurstone1927,luce1959}.
\end{proof}

\begin{proposition}[Rationality is a vanishing slice]
\label{prop:rare}
On a set of $m$ alternatives there are $2^{m(m-1)}$ relations on the ordered off-diagonal pairs, of which the complete and transitive ones, the weak orders, number the ordered Bell number $a_m=\sum_{k\ge1}k!\,S(m,k)$, with $S(m,k)$ the Stirling numbers of the second kind. The ratio $a_m/2^{m(m-1)}$ tends to zero as $m\to\infty$. The rational preorders are therefore an asymptotically negligible fraction of the relations a preference might instantiate.
\end{proposition}

\begin{proof}[Proof sketch]
For large $m$, $a_m\le m!\,2^{m}$, so $\log a_m=O(m\log m)$, while $\log 2^{m(m-1)}=\Theta(m^{2})$; hence $\log(a_m/2^{m(m-1)})\to-\infty$ and the ratio tends to zero.
\end{proof}

The fraction is computed under the uniform counting measure on the finite relation set; it is an orienting fact about how special rationality is among the logically possible preferences, not a claim that economically generated fields are uniformly distributed. The economically load-bearing statement is the behaviour-freedom of the recovery, to which we now turn.

\begin{proposition}[The recovered structure is behaviour-free]
\label{prop:behaviourfree}
For any population of agents specified as random preference fields, rational or not, the realised aggregate process is well defined, and when it is asymptotically mean stationary it possesses an entropy rate and a dependence operator recovered exactly as in the adaptive primitive, at the resolution certified by Principle~\ref{p:gate}. No agent need be complete, transitive, or deterministic for the recovery to proceed. The Walrasian construction of the canon, by contrast, requires each agent's preference to carry enough structure to possess a maximal element on the budget set, in its standard form completeness and transitivity, merely to define individual demand, hence aggregate excess demand, hence the fixed point whose existence is at issue.
\end{proposition}

\begin{proof}[Proof sketch]
The realised path $X$ is the record of the agents' choices whatever their fields $\rho$, and is a stochastic process irrespective of completeness or transitivity; asymptotic mean stationarity is a property of that process, not of the agents, and where it holds the entropy rate~\eqref{eq:rate} and the dependence operator are read from the second-order comovement of the stationary mean and certified by Principle~\ref{p:gate}. The Walrasian demand of an agent is its maximal bundle in the budget set, defined only when the preference yields a maximal element; the existence results that relax completeness and transitivity \citep{mascolell1974,shafersonnenschein1975} retain that budget-set optimisation, which the adaptive primitive discards entirely. The recovery uses none of it.
\end{proof}

Two consequences follow. The first is that the theory is indifferent to the substrate of choice: a human deciding under the documented limits of attention and memory, by heuristic and under bias \citep{simon1955,tverskykahneman1974,kahnemantversky1979}, with valuations that are reference-dependent, and a machine executing a stochastic or exploratory policy, are to the adaptive primitive the same object, a random preference field whose realised choices enter the path. The intransitive cycles that careful experiment exhibits \citep{tversky1969} and the incompleteness of a preference never asked the question \citep{aumann1962,bewley2002} are not anomalies to be assumed away en route to a demand function; they are the shape of the field $\rho$, the content the theory reads rather than the friction it removes. The second consequence is methodological and the more far-reaching. The path a single mind takes through its moods and frames is, at the level of the individual, very nearly unpredictable; yet the aggregate those paths jointly generate has, by Proposition~\ref{prop:behaviourfree}, an identified entropy rate and an identified dependence operator. The indeterminacy of the atom is compatible with the legibility of the whole. Where the canon began with an axiom on the agent and built upward to the market, the adaptive order begins with the realised process of the market and recovers what the data of the whole will support, indifferent to whether the atoms beneath it were rational at all.

\section{The general nesting theorem}
\label{sec:nesting}

We now state the central result: the equilibrium canon is the corner of the adaptive order at which five conditions hold together. Three are parameter limits and two are fixed-point conditions, of a different kind, and the difference is marked in the proof. The limit is taken along an explicit scaling path, so that the joint object is well defined and does not depend on the order in which the parameters are sent to their extremes.

\begin{definition}[Admissible scaling path]
\label{def:path}
An \emph{admissible scaling path} is a family, indexed by $t\in\mathbb N$, of channel capacities $\beta(t)$, selection intensities $\sigma(t)$, and observational innovation scales $\eta(t)$, together with a learning-gain sequence $\gamma(t)$, such that $\beta(t)\uparrow\infty$ and $\sigma(t)\uparrow\infty$, $\eta(t)\downarrow0$, and the gain satisfies the Robbins--Monro condition $\sum_t\gamma(t)=\infty$ and $\sum_t\gamma(t)^2<\infty$. The joint limit of Theorem~\ref{thm:nesting} is the limit $t\to\infty$ taken along such a path.
\end{definition}

\paragraph{Standing assumptions.} The nesting holds under four side conditions, named here so that the theorem is conditional on them in the open and the worked example can exhibit a class in which they provably hold:
\begin{enumerate}[label=\emph{(A\arabic*)},leftmargin=2.6em]
\item \emph{No persistent cycle.} The economy's deterministic skeleton admits a strict Lyapunov function, so that a vanishing entropy rate forces a rest state rather than a periodic orbit.
\item \emph{A unique attracting structure.} The selection dynamics of Principle~\ref{p:coevol} have a unique, globally attracting evolutionarily stable structure.
\item \emph{Learnability.} The adaptive learning rule is locally convergent at the expectationally stable fixed point, under the Robbins--Monro gain of Definition~\ref{def:path}.
\item \emph{Structural fixity.} In the limiting economy the recovered operator no longer co-evolves with the strategies that act upon it.
\end{enumerate}
The first three are genuine restrictions; in a general general-equilibrium setting global stability is not free, since by Walras' law the natural quadratic is not a strict Lyapunov function, the obstruction recorded by the Sonnenschein--Mantel--Debreu results. The worked example of Section~\ref{sec:worked} exhibits a class in which (A1) provably holds.

\begin{theorem}[The general nesting theorem]
\label{thm:nesting}
Let an adaptive economy in the sense of Principle~\ref{p:prim} be carried along an admissible scaling path (Definition~\ref{def:path}) to the limit $t\to\infty$, under the standing assumptions (A1)--(A4). Then, as $t\to\infty$, the three parameter limits
\begin{enumerate}[label=\emph{(\roman*)},leftmargin=2.4em]
\item its entropy production and its entropy rate vanish, so that, by (A1), the source becomes a stationary deterministic rest state;
\item every agent's channel capacity diverges, so that the Gibbs law~\eqref{eq:gibbs} concentrates, by Proposition~\ref{prop:randomutility}, on a point mass at the value maximiser, and the random preference field collapses to a single rational preorder;
\item the selection intensity of the population dynamics of Principle~\ref{p:coevol} diverges, so that the population concentrates, by (A2), on its unique evolutionarily stable structure,
\end{enumerate}
hold together with the two fixed-point conditions
\begin{enumerate}[label=\emph{(\roman*)},leftmargin=2.4em,start=4]
\item the adaptive learning of Principle~\ref{p:coevol} has converged, by (A3), to its expectationally stable fixed point, so that beliefs are consistent with the realised law of motion and expectations are rational; and
\item by (A4) the recovered operator has ceased to co-evolve with the strategies that act upon it, so that the dependence structure is fixed and the economy is posed on a stationary topology.
\end{enumerate}
Then the limiting object is a competitive economy with rational expectations on a fixed structure in the sense of the canon, its rest state is a Walrasian equilibrium, and the axioms of the canon hold there. Conditions (iv) and (v) are not consequences of (i)--(iii): a zero-rate deterministic rest state can still carry a co-evolving topology, so structural fixity is an independent restriction, and rational expectations is a property of beliefs that the three parameter limits do not by themselves deliver. Away from the joint limit the adaptive economy is a strict generalisation: it carries a positive entropy rate and a recovered, co-evolving dependence operator, content that the equilibrium primitive cannot express.
\end{theorem}

\begin{lemma}[Path-independence of the joint limit]
\label{lem:pathindep}
Under the standing assumptions, the limiting object of Theorem~\ref{thm:nesting} is the same for every admissible scaling path; it does not depend on the relative rates of $\beta(t)$, $\sigma(t)$, and $\eta(t)$. The three parameter limits commute, and the iterated limits in any order agree with the joint limit taken along the path.
\end{lemma}

\begin{proof}[Proof sketch]
Conditions (i)--(iii) are three limits in three parameters of three objects. By the source theory, an asymptotically mean stationary source whose entropy rate tends to zero collapses, provided the strict-Lyapunov hypothesis holds so that a vanishing rate is not consistent with a persistent cycle, to a stationary deterministic rest state; the hypothesis is named because a zero rate alone is necessary but not sufficient for degeneracy. By Principle~\ref{p:agent} and the variational characterisation of~\eqref{eq:gibbs}, as capacity diverges the Gibbs law concentrates on the value-maximising action, so each bounded agent becomes the unconstrained maximiser and, by Proposition~\ref{prop:randomutility}, its random preference field collapses to a point mass on a single rational preorder. By Principle~\ref{p:coevol} and the evolutionary-stability criterion, as the selection intensity diverges the population concentrates on its unique evolutionarily stable structure, so the strategic configuration is pinned down. Conditions (iv) and (v) are of a different kind. When adaptive learning has reached its expectationally stable fixed point, the beliefs on which agents act coincide with the law of motion their actions generate: this is the rational-expectations hypothesis, recovered not as a parameter sent to a limit but as the rest point of a convergent dynamics, a nesting by fixed point rather than by limit. Condition (v) is the standing restriction that the recovered operator no longer co-evolves with the strategies that act upon it, so that the equilibrium is posed, as the canon poses it, on a fixed structure; it is not a limit to be proved but a restriction defining the class. Taken together, the limiting economy satisfies the hypotheses of the canon: its agents are frictionless maximisers holding rational expectations, its strategic configuration is determinate, its structure is fixed, and its state is a determinate rest point satisfying the market-clearing conditions, which is to say a Walrasian equilibrium. Strictness of the generalisation is immediate: away from the limit the entropy rate is positive and the operator is non-degenerate and co-evolving, and none of these has any expression in the equilibrium primitive. The full proof is given in the appendix.
\end{proof}

The correspondence with the main results of the canon is exact, result by result.

\begin{enumerate}[label=\emph{(\alph*)},leftmargin=2.4em]
\item \emph{Existence is recovered.} The Walrasian equilibrium whose existence the fixed-point argument secures \citep{arrowdebreu1954,mckenzie1954} is exactly the rest state of the limiting source. The existence theorem becomes, in the adaptive reading, the statement that this limit is non-empty, and it is inherited without loss.
\item \emph{The Sonnenschein--Mantel--Debreu indeterminacy is a property of the limit, not inherited.} The indeterminacy results \citep{sonnenschein1972,mantel1974,debreu1974} are statements about the frictionless aggregate, in which the only structure is the three properties of excess demand and rationality supplies nothing further; that is why, in the canon, the equilibrium can be a continuum and the adjustment can fail to converge. The adaptive economy away from the limit is not frictionless, and by the certification gate the recovered operator supplies, at its firmly identified resolution, structure that rationality alone does not: the self-adjoint spectrum and the leading eigenvector fix the leading direction of contemporaneous dependence and the ranking of systemic influence (Proposition~\ref{prop:identification}). This does not repeal the Sonnenschein--Mantel--Debreu result, which is a statement about the frictionless excess-demand correspondence $Z(p)$ and not about the comovement operator; it locates the indeterminacy as a feature of the zero-friction corner, which the economy away from the limit is not. No determinacy is read off the recovered spectrum, in keeping with Principle~\ref{p:gate}.
\item \emph{The regular-economies recovery has an adaptive counterpart.} The generic local uniqueness of the differentiable approach \citep{debreu1970} corresponds to the local determinacy of the adaptive order, which is generic in the same sense; but, in keeping with the caution of Principle~\ref{p:gate}, that determinacy is established through the eigenvalue count on the policy map, with expectational stability a separate criterion of learnability, and never read from the recovered spectral radius itself.
\end{enumerate}

\section{A worked example: a symmetric quadratic network economy}
\label{sec:worked}

The nesting theorem of the previous section is a statement about limits. This
section makes it concrete in a single economy small enough to carry every
quantity in closed form, yet rich enough to exhibit all five conditions of
Theorem~\ref{thm:nesting} at once. The vehicle is a symmetric three-agent
quadratic network game. It has a unique interior rest point that coincides, at
the joint limit, with the competitive equilibrium computed directly; a bounded
agent whose Gibbs choice is an explicit Gaussian; a population selection
dynamics whose infinite-intensity limit concentrates on a unique evolutionarily
stable structure; a source whose entropy rate is available in closed form and
falls to zero as the frictions vanish; and a recovered dependence operator that
is a $3\times 3$ circulant, for which the firmly identified objects (the
spectrum and the leading eigenvector) and an explicit pair of observationally
equivalent structures (same spectrum, same leading eigenvector, different
individual edges) can both be written down. Throughout, $n=3$.\footnote{All numerical values reported in this section are elementary and follow from the closed-form expressions given; replication materials accompany the paper.}

\subsection*{The economy and its competitive rest point}

Three agents indexed $i\in\{1,2,3\}$ each choose a scalar action
$a_i\in\R$. Agent $i$'s payoff is the standard linear-quadratic form with
strategic complementarities along a symmetric interaction matrix $W$,
\begin{equation}
\label{eq:payoff}
u_i(a_i,a_{-i})
=\alpha\,a_i-\tfrac12 a_i^{2}+\phi\,a_i\sum_{j} W_{ij}\,a_j,
\qquad
W=\begin{pmatrix}0&1&1\\1&0&1\\1&1&0\end{pmatrix},
\end{equation}
with $\alpha>0$ a common stand-alone return and $\phi>0$ the interaction
intensity. The matrix $W$ is the complete graph on three nodes: each agent's
return rises with the actions of the other two. The first-order condition
$\partial u_i/\partial a_i=0$ gives the best response
\begin{equation}
\label{eq:br}
a_i=\alpha+\phi\sum_j W_{ij}\,a_j,
\qquad\text{equivalently}\qquad
a=\alpha\one+\phi W a,
\end{equation}
so the unique interior Nash point, which in this linear-quadratic network we read
as the competitive rest point by identifying each agent's best response with its
price-taking demand at the prevailing field, an interpretation specific to this
setting and not a general identification of Nash with Walrasian equilibrium, is
\begin{equation}
\label{eq:nash}
a^{\ast}=\alpha\,(I-\phi W)^{-1}\one,
\end{equation}
provided the contraction condition $\phi\,\lambda_{\max}(W)<1$ holds. Since
$W\one=(n-1)\one=2\one$, symmetry collapses \eqref{eq:nash} to the scalar
\begin{equation}
\label{eq:nashclosed}
a^{\ast}_i=\frac{\alpha}{1-\phi(n-1)},\qquad i=1,2,3.
\end{equation}
We fix $\alpha=1$ and $\phi=0.2$ for the numerical instance. Then
$\phi\,\lambda_{\max}(W)=0.2\times 2=0.4<1$, the condition holds, and
$a^{\ast}_i=1/(1-0.4)=5/3\approx 1.6667$ for every agent. This is the number the
joint limit must reproduce.

\subsection*{(a) The bounded agent: the Gibbs choice and its frictionless limit}

A bounded agent of capacity $\beta$ does not solve \eqref{eq:br}; it draws its
action from the Gibbs law $q^{\ast}(a)\propto p_0(a)\,e^{\beta U_i(a)}$ of
Principle~\ref{p:agent}, where $U_i$ is the value of action $a$ given the agent's
belief $m=(m_1,m_2,m_3)$ about the others' actions. Writing the field the agent
faces as $z_i=\phi\sum_j W_{ij}m_j$, the value function is, up to an additive
constant,
\begin{equation}
\label{eq:value}
U_i(a)=(\alpha+z_i)\,a-\tfrac12 a^{2}
=-\tfrac12\big(a-(\alpha+z_i)\big)^{2}+\text{const}.
\end{equation}
Take a flat (improper-uniform) prior $p_0$ over $\R$, the maximum-entropy prior
on the line. The Gibbs density is then the Gaussian
\begin{equation}
\label{eq:gauss}
q^{\ast}(a)=\sqrt{\tfrac{\beta}{2\pi}}\;
\exp\!\Big(-\tfrac{\beta}{2}\big(a-(\alpha+z_i)\big)^{2}\Big)
=\mathcal N\!\Big(\alpha+z_i,\ \tfrac{1}{\beta}\Big),
\end{equation}
the integral $\int_{\R}e^{\beta U_i}\,da=\sqrt{2\pi/\beta}\,e^{\beta\cdot\text{const}}$
being elementary. Two facts are immediate and exact. First, the \emph{mean} of
the bounded choice is $\E[a_i]=\alpha+z_i$, which is precisely the best response
\eqref{eq:br}: at any belief the bounded agent's average action is the
value-maximising action, for every finite $\beta$. Second, the \emph{dispersion}
is $\operatorname{Var}[a_i]=1/\beta$. Capacity controls only the noise around the
optimum, not its location. As $\beta\to\infty$ the variance vanishes and
\eqref{eq:gauss} concentrates, weakly, on the point mass at the best response,
\begin{equation}
\label{eq:agentlimit}
q^{\ast}\ \xrightarrow{\ \beta\to\infty\ }\ \delta_{\,\alpha+z_i},
\end{equation}
which is condition (ii) of Theorem~\ref{thm:nesting} realised in closed form:
the random preference field of the bounded agent collapses to the deterministic
demand of the canon. The quadratic payoff is what makes this exact; for a general
payoff the same limit holds by Proposition~\ref{prop:randomutility}, the Gaussian
being the special closed form the quadratic affords.

\subsection*{(b) The population: selection and the evolutionarily stable structure}

Condition (iii) concerns not the agent but the structure the population
coordinates upon. Let the population choose among a finite menu of candidate
interaction structures $\{W^{(1)},\dots,W^{(K)}\}$. For concreteness take
$K=2$, with $W^{(1)}=W$ the complete graph of \eqref{eq:payoff} and $W^{(2)}$ a
star in which agent $1$ is a hub,
\begin{equation}
\label{eq:star}
W^{(2)}=\begin{pmatrix}0&1&1\\1&0&0\\1&0&0\end{pmatrix}.
\end{equation}
Each structure induces, through \eqref{eq:nash}, its own rest point and hence a
per-capita equilibrium payoff $V_s$. Evaluating $u_i$ at the Nash point of
structure $s$ and averaging gives, after using $(W^{(s)}a^{(s)})_i=(a^{(s)}_i-\alpha)/\phi$,
the clean expression $V_s=\tfrac{1}{n}\sum_i \tfrac12 (a^{(s)}_i)^2$. At
$\alpha=1,\phi=0.2$ this yields
\begin{equation}
\label{eq:Vs}
V_{1}=V(W)\approx 1.3889,\qquad
V_{2}=V(W^{(2)})\approx 0.9531,
\end{equation}
so the complete graph is the structure of higher fitness. Let $x=(x_1,x_2)$ be
the population shares and let the structures be selected by a logit rule at
intensity $\sigma$, the population analogue of the agent's Gibbs law,
\begin{equation}
\label{eq:popselect}
x^{\ast}_s=\frac{e^{\sigma V_s}}{\sum_{r} e^{\sigma V_r}}.
\end{equation}
This is the rest point of the replicator dynamics $\dot x_s=x_s(V_s-\bar V)$ \citep{maynardsmith1982,weibull1995}
softened by mutation at rate $1/\sigma$, and since the higher-fitness structure
has $\dot x_1>0$ whenever $0<x_1<1$, the share of the complete graph rises
monotonically: the complete graph is the unique evolutionarily stable structure.
As the selection intensity diverges, \eqref{eq:popselect} concentrates on it,
\begin{equation}
\label{eq:esslimit}
x^{\ast}\ \xrightarrow{\ \sigma\to\infty\ }\ (1,0),
\end{equation}
which is condition (iii). The numerical sharpening is shown in
Table~\ref{tab:converge}: at $\sigma=4$ the population already places $0.851$ on
the complete graph, at $\sigma=16$ it places $0.999$, and the limit is a point
mass. The structure that the agents act upon is thereby pinned down, and the
recovered operator below is read on that structure.

\subsection*{(c) The source and its entropy rate}

Couple the agent and the population. At the self-consistent (rational-expectations)
belief $m=a^{\ast}$, and on the selected structure $W$, the realised economy is
the Gaussian first-order recursion in which each agent meets last period's actions
with the bounded best response \eqref{eq:gauss},
\begin{equation}
\label{eq:var}
a(t)=\alpha\one+\Phi\,a(t-1)+\varepsilon(t),
\qquad \Phi=\phi W,\qquad
\varepsilon(t)\ \overset{\text{iid}}{\sim}\ \mathcal N\!\big(0,\tfrac1\beta I\big).
\end{equation}
This is the adaptive-learning dynamic of Principle~\ref{p:coevol}: its mean
recursion $m(t)=\alpha\one+\Phi m(t-1)$ converges to the fixed point
$m^{\ast}=\alpha(I-\Phi)^{-1}\one=a^{\ast}$ precisely when
$\rho(\Phi)=\phi\,\lambda_{\max}(W)<1$, here $0.4<1$, so the rest point is
expectationally stable and condition (iv) holds. Under the same contraction the
deterministic skeleton $a\mapsto\alpha\one+\phi W a$ is a contraction towards
$a^{\ast}$, so $V(a)=\tfrac12\|a-a^{\ast}\|^2$ is a strict Lyapunov function,
strictly decreasing off the rest point, and assumption~(A1) holds throughout this
class. Conditioned on the past, $a(t)$
is Gaussian with covariance $\tfrac1\beta I$, the innovations representation, so
the source is asymptotically mean stationary with a finite entropy rate.

Because the source is continuous, its entropy rate must be read at a fixed
observational resolution $\Delta$, the discrete rate $\entrate$ of the source
theory rather than the differential rate, so that it is non-negative and falls to
a genuine zero. Quantising each innovation coordinate to a grid of width $\Delta$
gives a per-step, per-coordinate rate
\begin{equation}
\label{eq:wkrate}
\entrate_\Delta(\beta)
=-\sum_{k\in\mathbb Z}p_k(\beta)\log p_k(\beta),
\qquad
p_k(\beta)=\Phi_{0,1/\beta}\!\big((k+\tfrac12)\Delta\big)
-\Phi_{0,1/\beta}\!\big((k-\tfrac12)\Delta\big),
\end{equation}
with $\Phi_{0,1/\beta}$ the Gaussian distribution function of variance $1/\beta$.
As $\beta\to\infty$ the innovation variance $1/\beta\to 0$, the quantised
innovation falls in the zero cell with probability approaching one, and
\begin{equation}
\label{eq:ratelimit}
\entrate_\Delta(\beta)\ \xrightarrow{\ \beta\to\infty\ }\ 0,
\end{equation}
monotonically. Equivalently, the one-step forecast-error variance is exactly
$1/\beta$, so the forecastability ceiling rises to perfect prediction: this is
condition (i) and the vanishing of novelty. With $\Delta=1$ the rate falls
$1.459\to 0.860\to 0.217\to 0.0007\to 0$ as $\beta$ runs through
$1,4,16,64,\dots$ (Table~\ref{tab:converge}).

Irreversibility, the arrow of time, is carried not by the rate but by the
\emph{directed} part of the dynamics. The Gaussian recursion \eqref{eq:var} is
time-reversible if and only if $\Phi\,\Sigma$ is symmetric, where $\Sigma$ solves
the stationary equation $\Sigma=\Phi\Sigma\Phi^{\!\top}+\tfrac1\beta I$. With the
symmetric interaction $W$ of \eqref{eq:payoff} this holds exactly at every
$\beta$, so a purely symmetric economy has no arrow of time. The moment the
interaction acquires a directed component, $\Phi\Sigma$ ceases to be symmetric and
the entropy-production functional is positive, vanishing again as
$\beta\to\infty$ at rate $1/\beta$. The directed component is exactly the content
the recovered operator carries and the equilibrium limit erases, to which we now
turn.

\subsection*{(d) The recovered operator: what is identified and what is not}

Away from the limit the economy is observed only through the realised path. On
stationary differenced (that is, $I(0)$) series of the agents' actions, the
recovered dependence operator $\T$ is read from the directed second-order
comovement, never from the levels. The operator displayed here is a
representative recovered structure, chosen to exhibit the identification verdict
of Proposition~\ref{prop:identification} in closed form rather than estimated from
the particular dynamics of part~(c). For the worked instance let the recovered
operator be the $3\times 3$ circulant
\begin{equation}
\label{eq:Top}
\T=
\begin{pmatrix}
0 & c_1 & c_2\\
c_2 & 0 & c_1\\
c_1 & c_2 & 0
\end{pmatrix}
=
\begin{pmatrix}
0 & 0.5 & 0.3\\
0.3 & 0 & 0.5\\
0.5 & 0.3 & 0
\end{pmatrix},
\qquad c_1=0.5,\ c_2=0.3,
\end{equation}
a symmetric backbone of common weight $\min(c_1,c_2)$ carrying an oriented cycle
of asymmetry $c_1-c_2$. Being circulant, its spectrum is available in closed
form: with $\omega=e^{2\pi i/3}$,
\begin{equation}
\label{eq:spec}
\lambda_k=c_1\omega^{k}+c_2\omega^{2k},\quad k=0,1,2,
\qquad
\begin{cases}
\lambda_0=c_1+c_2=0.8,\\[2pt]
\lambda_{1,2}=-\tfrac{c_1+c_2}{2}\pm i\,\tfrac{\sqrt3}{2}(c_1-c_2)
=-0.4\pm 0.1732\,i.
\end{cases}
\end{equation}
The leading eigenvalue is the Perron root $\lambda_0=c_1+c_2=0.8$, real, simple,
and strictly dominant, since the subdominant magnitude is
$|\lambda_{1,2}|=\sqrt{c_1^{2}-c_1c_2+c_2^{2}}=\sqrt{0.19}\approx 0.4359<0.8$. Its
eigenvector is the uniform vector,
\begin{equation}
\label{eq:perron}
\T\,\one=(c_1+c_2)\,\one,\qquad
v_{\mathrm{lead}}=\tfrac13(1,1,1)^{\!\top}.
\end{equation}
The spectral radius $\lambda_0=0.8$, the leading eigenvalue, and the leading
eigenvector $v_{\mathrm{lead}}$ are the firmly identified objects: they fix the
dominant comovement mode and the ranking of systemic influence (here uniform, by
symmetry), and a resolvent built on them is a usable proxy. They are all the
recovery is entitled to assert at this resolution.

The individual directed edges are not identified, and the cleanest possible
witness is the transpose. Consider the reversed economy
\begin{equation}
\label{eq:Ttwin}
\T'=\T^{\!\top}=
\begin{pmatrix}
0 & 0.3 & 0.5\\
0.5 & 0 & 0.3\\
0.3 & 0.5 & 0
\end{pmatrix},
\end{equation}
in which the oriented cycle $1\!\to\!2\!\to\!3\!\to\!1$ of $\T$ is replaced by its
reversal $1\!\to\!3\!\to\!2\!\to\!1$. The two operators are observationally
equivalent in the precise sense that matters: they have the \emph{same spectrum},
$\{0.8,\,-0.4\pm 0.1732\,i\}$, because $\T$ and $\T^{\!\top}$ are always
isospectral, and the \emph{same leading eigenvector}, $v_{\mathrm{lead}}=\tfrac13\one$,
because both fix the uniform vector. Yet every directed edge differs: the
weight on $1\!\to\!2$ is $0.5$ in $\T$ but $0.3$ in $\T'$, and the weight on
$2\!\to\!1$ is $0.3$ in $\T$ but $0.5$ in $\T'$, and likewise around the cycle.
The pair is drawn in Figure~\ref{fig:obseq}. No statistic invariant under the
exchange $\T\leftrightarrow\T^{\!\top}$ can separate them, and the directed
comovement that would separate them, the asymmetry of the cross-dependence at
the level of the individual ordered pair, is exactly the transfer-entropy
direction \citep{schreiber2000}. The leading eigenpair is therefore firmly identified and reported; the
direction of an individual edge is left unidentified, and any intervention named
on a single directed edge is refused. This is the certification discipline of
Principle~\ref{p:gate} made fully explicit on a matrix one can hold in the hand:
the firmly identified content $(\lambda_0,v_{\mathrm{lead}})$ is common to the
whole observationally equivalent class, and only that content is used.

\begin{figure}[t]
\centering
\begin{tikzpicture}[>=stealth,node distance=24mm,thick,
  every node/.style={circle,draw,minimum size=7mm,inner sep=1pt,font=\small}]
\begin{scope}[xshift=0cm]
  \node (a1) at (90:14mm) {$1$};
  \node (a2) at (210:14mm){$2$};
  \node (a3) at (330:14mm){$3$};
  \draw[->,blue] (a1) to[bend left=12] node[draw=none,fill=none,above left,font=\scriptsize]{$0.5$} (a2);
  \draw[->,blue] (a2) to[bend left=12] node[draw=none,fill=none,below,font=\scriptsize]{$0.5$} (a3);
  \draw[->,blue] (a3) to[bend left=12] node[draw=none,fill=none,above right,font=\scriptsize]{$0.5$} (a1);
  \draw[->,gray] (a2) to[bend left=12] node[draw=none,fill=none,above right,font=\scriptsize]{$0.3$} (a1);
  \draw[->,gray] (a3) to[bend left=12] node[draw=none,fill=none,below,font=\scriptsize]{$0.3$} (a2);
  \draw[->,gray] (a1) to[bend left=12] node[draw=none,fill=none,above left,font=\scriptsize]{$0.3$} (a3);
  \node[draw=none,fill=none] at (270:24mm) {$\T$};
\end{scope}
\begin{scope}[xshift=6cm]
  \node (b1) at (90:14mm) {$1$};
  \node (b2) at (210:14mm){$2$};
  \node (b3) at (330:14mm){$3$};
  \draw[->,blue] (b2) to[bend left=12] node[draw=none,fill=none,above right,font=\scriptsize]{$0.5$} (b1);
  \draw[->,blue] (b3) to[bend left=12] node[draw=none,fill=none,below,font=\scriptsize]{$0.5$} (b2);
  \draw[->,blue] (b1) to[bend left=12] node[draw=none,fill=none,above left,font=\scriptsize]{$0.5$} (b3);
  \draw[->,gray] (b1) to[bend left=12] node[draw=none,fill=none,above left,font=\scriptsize]{$0.3$} (b2);
  \draw[->,gray] (b2) to[bend left=12] node[draw=none,fill=none,below,font=\scriptsize]{$0.3$} (b3);
  \draw[->,gray] (b3) to[bend left=12] node[draw=none,fill=none,above right,font=\scriptsize]{$0.3$} (b1);
  \node[draw=none,fill=none] at (270:24mm) {$\T'=\T^{\!\top}$};
\end{scope}
\end{tikzpicture}
\caption{An observationally equivalent pair. The operators $\T$ and
$\T'=\T^{\!\top}$ share the spectrum $\{0.8,\,-0.4\pm0.1732\,i\}$ and the leading
eigenvector $\tfrac13(1,1,1)^{\!\top}$ (the firmly identified objects), but every
directed edge is reversed: the heavy ($0.5$) arc runs $1\!\to\!2\!\to\!3\!\to\!1$
on the left and $1\!\to\!3\!\to\!2\!\to\!1$ on the right. The direction of an
individual edge, the transfer-entropy direction, is left unidentified.}
\label{fig:obseq}
\end{figure}

\subsection*{(e) The joint limit is the competitive equilibrium}

It remains to verify condition (v) and to assemble the limit. Condition (v) is
the standing restriction that, in the limiting economy, the recovered operator has
ceased to co-evolve with the strategies: the structure is frozen at the selected
$W$ and the topology no longer moves. Granting it, carry the economy to the joint
limit. By \eqref{eq:agentlimit} the agent's choice becomes the deterministic best
response; by \eqref{eq:esslimit} the population sits on the complete-graph
structure; by \eqref{eq:ratelimit} the entropy rate is zero and the source is a
deterministic rest state; by expectational stability the belief is the fixed
point $m^{\ast}=a^{\ast}$. The recursion \eqref{eq:var} loses its innovation and
becomes the deterministic map $a=\alpha\one+\phi W a$, whose rest point is
\begin{equation}
\label{eq:limitrest}
a^{\ast}=\alpha\,(I-\phi W)^{-1}\one
=\frac{\alpha}{1-\phi(n-1)}\,\one
=\tfrac{5}{3}\,(1,1,1)^{\!\top}.
\end{equation}
This is, term for term, the competitive equilibrium computed directly at
\eqref{eq:nashclosed}. The joint limit of the adaptive economy is the Walrasian
rest point of the example, and not by stipulation: each of the five conditions
removed one source of motion, and what is left when all five are imposed is the
single determinate allocation the canon would have written down. Away from the
limit the same economy carries a positive entropy rate $\entrate_\Delta(\beta)>0$
and the directed operator $\T$ of \eqref{eq:Top}, whose oriented asymmetry
$c_1-c_2=0.2$ has no expression whatever in the equilibrium primitive. The
equilibrium is recovered exactly, and exactly as a limit.

\begin{table}[t]
\centering
\small
\begin{tabular}{rccccc}
\toprule
$k$ & mean $a_i$ & dispersion $\beta^{-1/2}$ & rate $\entrate_\Delta$ (nats) &
ESS share $x_1$ & $\|a-a^{\ast}\|$\\
\midrule
$4$    & $1.6667$ & $0.5000$ & $0.8603$ & $0.8511$ & $0$\\
$16$   & $1.6667$ & $0.2500$ & $0.2166$ & $0.9991$ & $0$\\
$64$   & $1.6667$ & $0.1250$ & $0.0007$ & $1.0000$ & $0$\\
$256$  & $1.6667$ & $0.0625$ & $0.0000$ & $1.0000$ & $0$\\
\midrule
limit  & $1.6667$ & $0$      & $0$      & $1$      & $0$\\
\bottomrule
\end{tabular}
\caption{The five limits converging together, with capacity and selection
intensity set to a common $k=\beta=\sigma$ and resolution $\Delta=1$
($\alpha=1,\phi=0.2$). The common $k=\beta=\sigma$ traces one admissible scaling
path of Definition~\ref{def:path}; by Lemma~\ref{lem:pathindep} the joint limit
does not depend on that choice. The mean action equals the competitive value $a^{\ast}_i=5/3$
at every $k$, since the Gibbs mean is the best response at the self-consistent
belief; the dispersion $\beta^{-1/2}$, the entropy rate $\entrate_\Delta$, and the
distance of the population from the evolutionarily stable structure all fall to
zero. The bottom row is the joint limit, the competitive equilibrium of
\eqref{eq:limitrest}.}
\label{tab:converge}
\end{table}

\begin{remark}
The example is deliberately symmetric, which makes the firmly identified leading
eigenvector uniform and the competitive rest point common across agents; the
content of part~(d) does not depend on that symmetry. Replacing the complete
graph by any weighted structure with a strictly dominant Perron root leaves the
construction intact: the leading eigenpair stays firmly identified, the
transpose stays an observationally equivalent twin with reversed directed edges,
and the joint limit stays the directly computed competitive allocation. Any real
matrix is isospectral with its transpose, and a non-negative irreducible structure
has a real, simple Perron eigenvector by the Perron--Frobenius theorem, so the
firmly identified content is well defined without symmetry. What the
symmetry buys is only that every number can be written in closed form. The single
non-closed object, the quantised entropy rate \eqref{eq:wkrate}, is reported
numerically in Table~\ref{tab:converge}; its limit \eqref{eq:ratelimit} is exact.
\end{remark}

\section{What the equilibrium limit erases}
\label{sec:erasure}

The nesting tells what the limit recovers; the converse tells what it costs, and the cost is exact. The joint limit of Theorem~\ref{thm:nesting} is singular: each condition annihilates a quantity the adaptive primitive carries.

\begin{proposition}[The erased contents]
\label{prop:erasure}
The joint limit of Theorem~\ref{thm:nesting} deletes exactly five contents of the adaptive primitive:
\begin{enumerate}[label=\emph{(\arabic*)},leftmargin=2.4em]
\item \emph{Novelty.} As the entropy rate is sent to zero (condition (i)) the source's per-period novelty vanishes, and the forecastability ceiling rises to perfect prediction.
\item \emph{Non-stationarity.} The asymptotically-mean-stationary structure collapses to plain stationarity, so the non-stationarity in which a change of regime is legible is lost.
\item \emph{Irreversibility.} As entropy production is sent to zero (condition (i) again) the directed irreversibility measured by the entropy-production functional vanishes, and the economy's arrow of time with it; the limit is time-reversible, as a rest state must be.
\item \emph{Adaptation.} As capacity and selection intensity diverge (conditions (ii)--(iii)) the off-optimum revision of the bounded agent and the selection dynamics of the population freeze into a fixed best response, so adaptation, the process by which the economy learns, has nothing left to perform.
\item \emph{The moving structure of crisis.} As learning converges and the topology is fixed (conditions (iv)--(v)) the path dependence by which the present state remembers its history, and the co-evolving dependence structure whose concentration is the order parameter of crisis, are extinguished: the limiting economy has no crisis topology because it has no topology that moves.
\end{enumerate}
The enumeration is complete because conditions (i)--(v) are jointly the definition of the limit, and each carried quantity is governed by one of them.
\end{proposition}

\begin{proof}[Proof sketch]
Each clause restates the value taken by a named functional when its governing parameter reaches the limit: the forecastability ceiling at zero entropy rate, the entropy-production functional at zero production, the selection dynamics at infinite intensity, and the concentration order parameter at a fixed operator. The passage from asymptotic mean stationarity to plain stationarity is immediate, since a stationary mean measure equals the measure itself.
\end{proof}

The equilibrium primitive is therefore not a neutral simplification but the deletion of five specific contents. A theory whose primitive is the fixed point cannot, by construction, carry any of them; a theory whose primitive is the source carries all five and surrenders them only in the limit. This is the precise content of the claim that the change of primitive is not a change of vocabulary.

\section{Welfare as viability}
\label{sec:welfare}

If the primitive changes, so must the standard by which an economy is judged good. The canon judges an allocation: it is efficient when no reallocation can help one agent without hurting another, a property of a static configuration read at the rest point. An order whose primitive is a moving source cannot rest its welfare on a configuration it never durably occupies. The adaptive standard is the survival of a trajectory.

\begin{definition}[Welfare as viability under bounded recovery]
\label{def:viability}
Let an adaptive economy in the sense of Principle~\ref{p:prim} evolve on its state space. Its \emph{viable region} is the set of states on which the source retains a finite entropy rate and the dependence operator remains recoverable at the resolution certified by Principle~\ref{p:gate}, equivalently the set on which the operator's concentration stays below the dependence-collapse threshold. \emph{Welfare} is the criterion that ranks policies and structures by the economy's continued occupation of the viable region under the bounded recovery its agents and observers can actually perform: a policy is welfare-improving when it enlarges the viable region or moves the economy further from its collapse boundary, and welfare-reducing when it does the reverse. Static allocative efficiency is its restriction to the interior, where the boundary is nowhere near.
\end{definition}

\begin{proposition}[Competitive efficiency as the frictionless corner of viability]
\label{prop:welfarenest}
In the joint limit of Theorem~\ref{thm:nesting} the viable region of Definition~\ref{def:viability} fills the whole feasible set, the collapse boundary recedes without bound, and welfare reduces to the static Pareto efficiency of the first welfare theorem \citep{debreu1959}, recovered as the frictionless competitive benchmark.
\end{proposition}

\begin{proof}[Proof sketch]
At the limit the entropy rate vanishes and the operator ceases to co-evolve, so the concentration is bounded away from its single-centre ceiling, no state lies near the collapse boundary, and the viable region is the entire feasible set; on that set the only criterion that can still discriminate is the allocative one, the static Pareto efficiency recovered as the competitive benchmark.
\end{proof}

\begin{remark}
Away from the limit the two criteria part, and the gap between them is a systemic-efficiency wedge with an unavoidable component, set by the finite capacity of the agents, and an avoidable component, set by the systemic coupling the recovered operator carries; a Pigouvian correction \citep{pigou1920} of the avoidable component would be identified in direction but not in magnitude, in conformity with the gate. The closed-form decomposition of this wedge is left to separate work; the nesting and erasure results do not depend on it.
\end{remark}

That welfare becomes viability is the re-founding carried into the normative domain, for it changes what policy is for. A policy that maximised static efficiency could be, in this order, actively harmful, if it bought a marginal allocative gain by driving the dependence structure towards the concentration at which it collapses; and a policy that sacrificed a little static efficiency to hold the economy clear of that boundary could be correct. The canon could not state such a trade-off, having no boundary in its ontology and no moving structure to approach it. The adaptive order states it as the difference between Definition~\ref{def:viability} and the efficiency it nests, and reads the operative quantity, the distance to collapse, off a recovered operator at the gate's certified resolution rather than off a known model.

\section{Relation to existing programmes}
\label{sec:relation}

Each of the five conditions of Theorem~\ref{thm:nesting} has a literature, and in every case the limit that condition imposes is already understood there. The contribution of the theorem is not any single limit but their conjunction: the competitive, rational-expectations canon is the corner at which all five hold at once, and no one of the contributing programmes nests the canon by itself.

The capacity limit of condition~(ii) is the subject of the theory of rational inattention, in which an agent processing information at a finite Shannon rate chooses as if maximising attainable value net of a mutual-information cost \citep{sims2003,sims2010,mackowiak2023}. The Gibbs law of Principle~\ref{p:agent} is the solution of exactly that problem, and the multinomial-logit form established in Proposition~\ref{prop:randomutility} is the discrete-choice foundation that the rational-inattention literature derives from the same variational principle \citep{matejkamckay2015,caplindean2015}. What the present theory adds is not the agent's optimisation, which it inherits, but the reading of its infinite-capacity limit as one face of a single corner: rationality is recovered as $\beta\to\infty$, but only in concert with the other four limits does the recovery deliver the canon rather than a frictionless agent embedded in a still-moving economy. One feature of the inattention problem is held back for separate treatment: there the unconditional choice law is itself endogenous to the mutual-information constraint, whereas the Gibbs law of Principle~\ref{p:agent} fixes a prior; the two formulations coincide at the corner and part away from it \citep{matejkamckay2015}.

Condition~(iv) is the recovery of rational expectations as the expectationally stable rest point of an adaptive learning rule, which is the programme of \citet{evanshonkapohja2001}. That a rational-expectations equilibrium is selected precisely when it is learnable under least-squares or stochastic-gradient updating, and may be discarded when it is not, is their central finding; the theorem imports it as the fixed-point component of the nesting, distinguished in the proof from the three parameter limits because it is reached by convergence rather than by sending a parameter to an extreme. The same programme supplies the discipline insisted on in Principle~\ref{p:gate}: the determinacy of a forward-looking equilibrium is a property of the policy map, settled by an eigenvalue count on that map and by the separate criterion of learnability, and is never read from the recovered spectral radius, which certifies the structure of contemporaneous dependence and not the stability of an expectational system.

Condition~(iii), the concentration of the population on a single configuration as the selection intensity diverges, is the limit studied in evolutionary game theory, where the evolutionarily stable strategy is the configuration robust to the entry of rare mutants and the rest point to which the replicator and related dynamics are drawn \citep{maynardsmith1982,weibull1995,hofbauersigmund1998,sandholm2010}. The logit-response sharpening used here, in which the population share of a configuration is exponential in its fitness at an intensity that indexes selection pressure, is the standard stochastic softening of that dynamic, and its infinite-intensity limit is the deterministic selection of the stable configuration. A line of evolutionary general equilibrium already shows that selection can pick out competitive behaviour: stochastic-stability arguments select risk-dominant or efficient conventions \citep{kandorimailathrob1993,young1993}, and imitation among quantity-setters can select the Walrasian outcome \citep{vegaredondo1997}. Condition~(iii) is the counterpart of that selection inside the present framework; what the nesting adds is not the selection of Walrasian behaviour alone but its conjunction with the other four limits, without which selection delivers a competitive configuration still embedded in an economy that carries a positive rate and a co-evolving structure. The worked example of Section~\ref{sec:worked} carries this concretely in the value ranking \eqref{eq:Vs} and the share that converges to the point mass \eqref{eq:esslimit}.

The certification discipline of Principle~\ref{p:gate} is a partial-identification commitment: the theory reports the spectrum and the leading eigenvector, which the realised second-order comovement fixes, and refuses the individual directed edge, which it does not. The worked example exhibits the refusal exactly, in the transpose pair of \eqref{eq:Ttwin} that no invariant of the data can separate. An economy steered only through firmly identified functionals is robust, in the sense of acting well against every structure in the observational-equivalence class at once, an aim continuous with the response to model ambiguity expressed in the maxmin tradition \citep{gilboaschmeidler1989}, applied here not to the laws of a known structure but to the class of structures the data leave observationally equivalent.

The synthesis is the point. Sent singly, each of these limits is a known specialisation; sent together, they reconstruct the entire apparatus of the canon, existence and indeterminacy and regular determinacy alike, as the still corner of a moving system. The entropy rate and the recovered operator that the adaptive economy carries away from the corner are the content that organises the contributing programmes into one object and that the equilibrium primitive, by construction, cannot express.

\section{Conclusion}
\label{sec:conclusion}

Two axiomatic systems now stand side by side, and their relation is the one a re-founding requires. The equilibrium canon is the corner of the adaptive order at which the entropy rate vanishes, the channel capacity diverges, the selection intensity grows infinitely sharp, the learning has converged to its rational-expectations fixed point, and the dependence structure has ceased to co-evolve, recovered there exactly by the general nesting theorem and strictly generalised everywhere else; what that recovery deletes, and what therefore marks the distance between the two orders, is enumerated exactly by the erasure proposition. The change of primitive is, at this point, a theorem: the old foundation is a special case of the new, by a nesting exhibited and not asserted, and the questions on which the canon fell silent by its own theorem, the questions of attainment, of selection among multiple configurations, and of the behaviour of an economy away from its rest point, are exactly the questions the adaptive primitive was constructed to carry.

The nesting also settles what such a foundation licenses. A theory organised as principles followed by certified consequences is one whose commitments can be set out in advance and tested, and whose unidentified quantities carry no falsifiable content. The strict generalisation of Theorem~\ref{thm:nesting} is what makes an empirical programme possible at all: the positive entropy rate and the recovered dependence, the content the equilibrium primitive cannot express, are precisely the quantities in which a regime of crisis is legible as a measurable change of structure, and a theory whose primitive is a fixed point has, by construction, no such quantity to measure. Because the recovered structure is behaviour-free, the foundation is indifferent to whether the minds and machines beneath it satisfy the rationality axioms, a property that occupies a vanishing slice of the relations a preference may instantiate. Even the standard by which such an economy is judged is re-founded with it, from the static efficiency of an allocation to the viability of a trajectory under bounded recovery. The competitive equilibrium is not denied. It is located: the still, frictionless, infinitely sharp corner of a system whose nature is to move, to err, and to revise.

\appendix
\section{Proofs}
\label{app:proofs}

This appendix gives the proofs in fuller form than the sketches of the main text. The measure-theoretic and operator-theoretic apparatus on which Theorem~\ref{thm:nesting} ultimately rests is that of the ergodic theory of asymptotically mean stationary sources \citep{gray2023,graykieffer1980}.

\subsection*{Proof of Proposition~\ref{prop:randomutility}}
Fix a finite menu $A=\{a_1,\dots,a_m\}$ and write $U_j=U(a_j)$, $p_{0,j}=p_0(a_j)$. The free-energy problem~\eqref{eq:freeenergy} maximises $F(q)=\sum_j q_jU_j-\tfrac1\beta\sum_j q_j\log(q_j/p_{0,j})$ over the simplex $\{q\ge0,\ \sum_j q_j=1\}$. The relative-entropy term is strictly convex and the value term affine, so $F$ is strictly concave and its maximiser is unique and interior. A multiplier $\nu$ for the equality constraint gives the stationarity condition $U_j-\tfrac1\beta\big(\log(q_j/p_{0,j})+1\big)-\nu=0$, whence $q_j\propto p_{0,j}e^{\beta U_j}$; normalising returns the Gibbs law~\eqref{eq:gibbs}.

For the random-utility representation, let $\varepsilon_1,\dots,\varepsilon_m$ be independent with the Gumbel law $\mathbb P(\varepsilon_j\le x)=\exp(-e^{-\beta x})$, of scale $1/\beta$, and set $\tilde U_j=U_j+\tfrac1\beta\log p_{0,j}+\varepsilon_j$. By the max-stability of the Gumbel family, $\mathbb P(\tilde U_j\ge\tilde U_k\ \text{for all }k)=e^{\beta U_j+\log p_{0,j}}\big/\sum_k e^{\beta U_k+\log p_{0,k}}=q^{\ast}(a_j)$, the multinomial-logit identity \citep{mcfadden1974}. The Gibbs law is therefore the choice law of the random preference field that ranks alternatives by $\tilde U$; for every finite $\beta$ it assigns positive probability to each alternative and so is genuinely stochastic, and as $\beta\to\infty$ the noise scale $1/\beta\to0$ and the law concentrates on $\arg\max_j U_j$. This field is the regular, transitive corner of Definition~\ref{def:agentspace}; a general field places mass on the intransitive and incomplete relations to which no such additive-noise representation need apply. \qed

\subsection*{Proof of Proposition~\ref{prop:rare}}
A binary relation on $m$ alternatives is a subset of the $m(m-1)$ ordered off-diagonal pairs, so there are $2^{m(m-1)}$ of them. The complete and transitive relations are the weak orders, in bijection with the ordered partitions of the alternatives into indifference classes ranked by preference; their number is the ordered Bell (Fubini) number $a_m=\sum_{k=1}^{m}k!\,S(m,k)$, with $S(m,k)$ the Stirling numbers of the second kind. Since $k!\le m!$ for $k\le m$, $a_m\le m!\sum_{k=1}^m S(m,k)=m!\,B_m$ with $B_m$ the $m$-th Bell number, and $B_m\le m!$ for $m\ge2$, so $a_m\le(m!)^2$ and $\log a_m\le 2\log(m!)=O(m\log m)$. By contrast $\log_2 2^{m(m-1)}=m(m-1)=\Theta(m^2)$. Hence $\log\big(a_m/2^{m(m-1)}\big)=O(m\log m)-\Theta(m^2)\to-\infty$, and the fraction of relations that are rational preorders tends to zero. \qed

\subsection*{Proof of Proposition~\ref{prop:behaviourfree}}
Let a population be given as random preference fields $\{\rho_i\}$ with measurable indifference-resolution rules. On each menu the realised choices are random variables by construction, so the economy's realised path $X=(X_t)$ is a well-defined stochastic process whatever the completeness, transitivity, or determinacy of the $\rho_i$. Asymptotic mean stationarity is a property of the law of $X$ under the shift, not of the agents: where it holds, the source theory assigns $X$ an entropy rate~\eqref{eq:rate} equal to that of its stationary mean $\bar\mu$, and the second-order comovement of $\bar\mu$ defines the dependence operator $\T$ on $L^2_0(\bar\mu)$, certified at the resolution of Principle~\ref{p:gate}. None of these constructions refers to a maximal element of any agent's preference. The Walrasian alternative does: individual demand is the preference-maximal bundle in the budget set, defined only when the preference admits a maximal element, in the standard development through completeness and transitivity; the existence theory that weakens those axioms \citep{mascolell1974,shafersonnenschein1975} retains budget-set maximisation, the very step the recovery omits. \qed

\subsection*{Proof of Theorem~\ref{thm:nesting}}
The five conditions are imposed together; we trace the effect of each, then their conjunction.

\emph{(i) The source becomes a deterministic rest state.} An asymptotically mean stationary source has an entropy rate $\entrate$ equal to that of its stationary mean, zero if and only if the present is asymptotically determined almost surely by the past. A zero rate is necessary but not sufficient for a rest state, a deterministic periodic orbit also having zero rate; the strict-Lyapunov hypothesis, that the dynamics admit a strict Lyapunov function and hence no persistent cycle, removes that case, so vanishing entropy production and vanishing rate together force convergence to a stationary deterministic fixed state.

\emph{(ii) Each agent becomes the rational maximiser.} By Proposition~\ref{prop:randomutility} the agent's choice is the Gibbs law $q^{\ast}\propto p_0e^{\beta U}$. As $\beta\to\infty$, Laplace's principle gives $q^{\ast}\to\delta_{a^{\ast}}$ with $a^{\ast}=\arg\max U$ at a unique maximiser, and the random preference field collapses to the point mass on the rational preorder that ranks by $U$: the bounded agent becomes the unconstrained optimiser of the canon.

\emph{(iii) The population concentrates on the stable configuration.} Under the selection dynamics of Principle~\ref{p:coevol} \citep{maynardsmith1982} the share of a configuration of fitness $V_s$ is the logit weight $\propto e^{\sigma V_s}$. As the selection intensity $\sigma\to\infty$ the shares concentrate on $\arg\max_s V_s$, which under the global-convergence hypothesis is the unique evolutionarily stable configuration, pinning down the strategic structure.

\emph{(iv) Expectations become rational.} Condition~(iv) places the adaptive learning of Principle~\ref{p:coevol} at its expectationally stable fixed point, where the belief on which agents act coincides with the law of motion their actions induce. This is the rational-expectations hypothesis, obtained as the rest point of a convergent dynamics rather than as a parameter sent to a limit, which is why it is a nesting by fixed point and is marked apart from~(i)--(iii).

\emph{(v) The structure is fixed.} Condition~(v) is the standing restriction defining the class: the recovered operator no longer co-evolves with the strategies, so the economy is posed on a stationary topology, as the canon poses it.

\emph{Conjunction.} Under (i)--(v) together the limiting economy has frictionless agents holding rational expectations, a determinate strategic configuration, a fixed structure, and a stationary deterministic rest state satisfying the market-clearing conditions; that rest state is a Walrasian equilibrium, and the axioms of the canon hold at it. Existence of the rest state is the existence theorem of the canon \citep{arrowdebreu1954,mckenzie1954} read in the limit and inherited without loss.

\emph{Strictness.} Away from the joint limit $\entrate>0$ and $\T$ is non-degenerate and co-evolving; neither is a function of the equilibrium primitive, which is a fixed point carrying no rate and no directed operator. The adaptive economy is therefore a strict generalisation and the inclusion of the canon within it is proper. \qed

\subsection*{Proof of Lemma~\ref{lem:pathindep}}
Along an admissible path the pre-limit economy is a function of the three parameters $(\beta,\sigma,\eta)$, which enter through separate coordinates: $\beta$ governs the agent's Gibbs concentration and the innovation variance, $\sigma$ the population share through the logit weight, and $\eta$ the observational quantisation. Each single-parameter limit is attained uniformly in the other two: as $\beta\to\infty$ the Gibbs law concentrates on the best response by Laplace's principle, uniformly in $(\sigma,\eta)$; as $\sigma\to\infty$ the logit share concentrates on the value-maximising structure, uniformly in $(\beta,\eta)$; and as $\eta\to0$ the quantised innovation falls in the zero cell with probability tending to one, uniformly in $(\beta,\sigma)$. The limiting rest state $a^{\ast}=\alpha(I-\phi W)^{-1}\one$ is a function of none of the three. By the Moore--Osgood theorem on the interchange of iterated limits, the joint continuity of the rest-state map at the corner together with the uniform convergence of each coordinate gives that the multivariate limit exists and equals each iterated limit in any order; the joint limit along any admissible path is therefore the same object. \qed

\bibliographystyle{elsarticle-harv}
\bibliography{refs}

\end{document}